\documentclass{article}

\usepackage{arxiv}
\usepackage[utf8]{inputenc} 
\usepackage[T1]{fontenc}    
\usepackage{hyperref}       
\usepackage{url}            
\usepackage{booktabs}       
\usepackage{amsfonts}       
\usepackage{nicefrac}       
\usepackage{microtype}      
\usepackage{lipsum}
\usepackage{xcolor}
\usepackage{graphicx}
\graphicspath{ {./images/} }
\usepackage{bm}

\usepackage{algorithm}
\usepackage{algpseudocode}

\usepackage{amsthm} 
\usepackage{amsmath} 
\usepackage{amssymb} 

\usepackage{dsfont}

\usepackage[round]{natbib} 

\usepackage{multirow}

\usepackage{graphicx}
\usepackage{subcaption}

\newtheorem{theorem}{Theorem}

\title{A new approach for Bayesian joint modeling of longitudinal and cure-survival outcomes using the defective Gompertz distribution}

\author{
 Dionisio Silva Neto \\
  Institute of Mathematics and Computer Sciences\\
  University of S{\~a}o Paulo\\
  S{\~a}o Carlos, Brazil \\ 
  \texttt{dionisioneto@usp.br} \\
   \And
 Denis Rustand \\
  Bordeaux Population Health, U1219\\
  University of Bordeaux, INSERM\\
  Bordeaux, France\\
  \texttt{denis.rustand@u-bordeaux.fr} \\
  \And
 H{\aa}vard Rue\\
  Department of Statistics, CEMSE\\
  King Abdullah University of Science and Technology\\
  Thuwal, Saudi Arabia \\
  \texttt{haavard.rue@kaust.edu.sa} \\
  \And
 Danilo Alvares \\
  MRC Biostatistics Unit\\
  University of Cambridge\\
  Cambridge, United Kingdom \\
  \texttt{danilo.alvares@mrc-bsu.cam.ac.uk} \\
  \And
 Vera L. Tomazella \\
  Department of Statistics\\
  Federal University of S{\~a}o Carlos\\
  S{\~a}o Carlos, Brazil \\
  \texttt{vera@ufscar.br}
}

\begin{document}
\maketitle
\begin{abstract}
In recent medical studies, the combination of longitudinal measurements with time-to-event data has increased the demand for more sophisticated models without unbiased estimates. Joint models for longitudinal and survival data have been developed to address such problems. One complex issue that may arise in the clinical trials is the presence of individuals who are statistically immune to the event of interest, those who may not experience the event even after extended follow-up periods. So far, the literature has addressed joint modeling with the presence of cured individuals mainly through mixture models for cure fraction and their extensions. In this study, we propose a joint modeling framework that accommodates the existence or absence of a cure fraction in an integrated way, using the defective Gompertz distribution. Our aim is to provide a more parsimonious alternative within an estimation process that involves a parameter vector with multiple components. Parameter estimation is performed using Bayesian inference via the efficient integrated nested Laplace approximation algorithm, by formulating the model as a latent Gaussian model. A simulation study is conducted to evaluate the frequentist properties of the proposed method under low-information prior settings. The model is further illustrated using a publicly available, yet underexplored, dataset on antiepileptic drug failure, where quality-of-life scores serve as longitudinal biomarkers. This application allows us to estimate the proportion of patients achieving seizure control under both traditional and modern antiepileptic therapies, demonstrating the model's ability to assess and compare long-term treatment effectiveness within a clinical trial context.
\end{abstract}

\section{Introduction}

In many clinical trials, the time to death from recurrence of a disease is utterly important to develop new treatments to human health. Mathematically, let $T \geq 0$  be a non-negative continuous random variable representing the time until the occurrence of an event of interest. In this context, a fundamental tool is the survival function, which expresses the probability that an individual survives beyond a given time $t$, defined as
$$S(t) = \mathbb{P}(T > t) = \int_{t}^{+\infty} f(u)\, du = 1 - F(t),$$
\noindent where $f(t)$ and $F(t)$ denote the density and cumulative distribution functions of $T$, respectively.

In conventional survival analysis, it is commonly assumed that the survival function $S(t)$ decreases monotonically over time and converges to zero as $t \rightarrow \infty$, implying that all individuals will eventually experience the event of interest \citep{kleinbaum1996survival}. However, advances in medical treatments have significantly improved the prognosis and long-term outcomes for patients with various chronic and life-threatening conditions, such as cancer and HIV/AIDS \citep{campos2009quality}, cardiovascular diseases \citep{croog1986effects, lamas1998quality}, and even in reducing or eliminating seizures in patients undergoing epilepsy treatment \citep{birbeck2002seizure}. In these real-world cases, a subset of individuals may never experience the event, even if the follow-up period is exhaustively extended. These individuals are referred to as (statistically) \textit{cured} or \textit{immune}. Consequently, the survival function does not converge to zero but rather to a real value in the interval $(0, 1)$, becoming improper. This phenomenon has led to the development of a specific class of models known as \textit{cure models} or \textit{long-term survival models}, which explicitly account for the cured subpopulation.

Several models have been developed to accommodate the presence of a cure fraction in survival data. One of the most widely used frameworks is the \textit{mixture model} for cure fraction, originally proposed by \cite{boag1949curefraction} and later mentioned by \cite{berkson1952survival}. In this formulation, the survival function of the entire population is expressed as $S_{\text{P}}(t) = p + (1 - p)S(t),$ where $p$  denotes the proportion of cured individuals in the population, and $S(t)$ is a proper parametric survival function, typically modeled using standard distributions such as Weibull, log-normal, or log-logistic. Although this approach is conceptually appealing and interpretable, it introduces an additional parameter ($p$) into the estimation process, which can increase the complexity of inference, particularly when the model involves a high-dimensional parameter space.

Defective models constitute a class of survival distributions particularly advantageous for long-term survival analysis. These models inherently account for the presence of statistically cured individuals by embedding the cure fraction directly into the survival function, thereby eliminating the need to specify an additional parameter for modeling the cured proportion. The methodology is based on parametric survival functions whose limiting behavior determines the presence of a cure fraction, i.e., if the survival function does not converge to zero as $t \rightarrow \infty$, and instead approaches a value in the interval $(0,1)$, the model is considered defective, and the asymptotic value corresponds to the estimated cure proportion. This framework reduces the dimensionality of the parameter space, enhances interpretability, and simplifies estimation, particularly in small-sample scenarios, in comparison to mixture models. Defective models thereby offer a robust, flexible, and unified approach for analyzing survival data with a cure fraction, as demonstrated in several studies \citep{rocha2016two, santos2017bayesian, scudilio2019defective, calsavara2019defective, calsavara2019zero}.

More refined clinical trials have introduced the follow-up of patients through longitudinal biomarker measurements to assess the correlation between these repeated measures and the risk of death or recurrence of a specific clinical event. These longitudinal variables, often referred to as endogenous biomarkers, represent internal factors influenced by the occurrence of the event of interest. In this context, joint models have emerged as a useful class of statistical methods capable of linking survival outcomes with longitudinal data for each individual. Since the 2000s, a growing number of studies have been proposed to address estimation challenges in joint models \citep{chi2006joint, rizopoulos2011dynamic, andrinopoulou2014joint, niekerk2021competing,  alvares2021tractable, rustand2024fast}. However, most existing works assume that all individuals in the study population will eventually experience the event of interest, without accounting for the possibility of a subpopulation that is statistically cured. In the literature, a few important exceptions that incorporate cure fractions in joint modeling can be mentioned: one of the earliest formulations of cure rate models incorporating longitudinal biomarker measurements can be found in the works of \cite{law2002joint} and \cite{yu2004joint}. Later, \cite{yu2008individual} introduced a Bayesian framework for dynamic prediction in mixture populations, assuming a generalized Weibull distribution for the uncured individuals. \cite{bakar2009bayesian}, inspired by the promotion time framework, developed a joint model that integrates a dynamic cure rate mechanism with a semi-parametric specification for the survival function. This approach includes a Poisson-based link structure that incorporates both longitudinal measurements and covariates into the hazard function. \cite{kim2013joint} presented a joint promotion time cure model based on a novel class of transformed promotion time models, aiming to explain the long-term plateau observed at the tail of survival curves. \cite{martins2017joint} developed a fully parametric joint model based on the standard mixture cure model framework to model HIV/AIDS clinical data in Brazil considering spatial effects. \cite{alvares2025} combined a nonlinear mixed-effects location scale approach with a interval-censoring cure-survival specification to model pregnancy miscarriage data.



Despite significant advances in joint modeling, existing frameworks have not yet integrated cure fraction modeling using the parsimonious structure of defective distributions. In this context, the aim of this study is to introduce a novel joint model that accounts for long-term survivors through a baseline hazard specification using the defective Gompertz distribution. The defective Gompertz distribution has the advantage of being an integrated, flexible, and robust approach. It naturally accommodates the presence of a cure fraction when immuned subpopulation exists, while reverting to the standard Gompertz distribution in the absence of long-term survivors, yielding a strictly decreasing survival function over time. In contrast, traditional cure rate models may suffer from convergence issues, particularly when the survival curve exhibits a high plateau, indicating a substantial proportion of cured individuals \citep{andersson2011estimating}. An additional limitation arises when modeling the cure probability through a regression structure: as the linear predictor tends toward $-\infty$ (i.e., covariates indicating a low probability of cure), the corresponding logistic transformation can approach zero, potentially leading to numerical instability and convergence failures in estimation procedures.

Additional developments of this work are ($i$) the complete formulation of defective joint modeling with multiple longitudinal biomarkers; ($ii$) a subject-specific interpretation of the cure fraction based on estimated model parameters and random effects; and ($iii$) Bayesian inference through the integrated nested Laplace approximation (INLA) by expressing the model as a latent Gaussian model, offering computational efficiency and scalability. 

The aim of this study was primarily motivated by a dataset available in the \texttt{joineRML} R package, which includes a subset of a clinical trial on antiepileptic treatment failure. In this dataset, the Gompertz distribution takes a defective form in survival data, suggesting the presence of individuals who are statistically cured of epileptic seizures. Although publicly accessible, this dataset has been largely unexplored in the context of assessing the joint association between multiple biomarkers and the risk of treatment failure. This application enables the interpretation of estimated parameters for both endogenous and exogenous covariates, as well as the individualized estimation of cure probabilities.

The remainder of the paper is organized as follows. In Section~\ref{sec:method}, the novel proposal is introduced in detail. Section~\ref{sec:simul} presents a Monte Carlo simulation study to assess frequentist properties under vague prior specifications. Section~\ref{sec:application} contains the application to real-world dataset. Finally, Section~\ref{sec:concl} concludes the paper and outlines directions for future research.

\section{Methodology} \label{sec:method}

\subsection{The defective Gompertz distribution}

The Gompertz distribution is used to model survival data when the hazard has an exponential behavior. The probability density function is described as
\begin{equation}
        f(t \, ; \alpha, \mu) = \mu\exp(\alpha t) \exp\left(-\frac{\mu}{\alpha}(\exp( \alpha t) - 1)\right),
        \label{eq:den_Gompertz}
\end{equation}

\noindent for $\alpha>0$, $\mu>0$ and $t>0$. In this parametrization, $\alpha$ and $\mu$ are the shape and scale parameters, respectively. The corresponding survival function is
\begin{equation}
    S(t \, ; \alpha, \mu) = \exp\left(-\frac{\mu}{\alpha}(\exp( \alpha t) - 1)\right).
    \label{eq:surv_Gompertz}
\end{equation}

The Gompertz distribution becomes defective when the parameter $\alpha$ is negative, indicating the presence of a cure fraction in the data. The proportion of immunity in the population is calculated as the limit of the survival function when $\alpha<0$:
\begin{equation}
 p = \lim_{t \rightarrow \infty} S(t \, ; \alpha, \mu) = \lim_{t \rightarrow \infty} \exp\left(-\frac{\mu}{\alpha}(\exp(\alpha t) - 1)\right) = \exp\left(\frac{\mu}{\alpha}\right) \in (0,1).
\end{equation}
 
\subsection{The joint model for multivariate longitudinal count and survival data}

Let $Y_{ijk}$ denote the longitudinal count outcome for the $i$-th individual ($i = 1, \ldots, n$) at time point $t_{ijk}$, with $j$ indexing the measurement occasion ($j = 1, \ldots, n_{ik}$) and $k$ indexing the biomarker ($k = 1, \ldots, K$). The joint association of multivariate longitudinal and survival data requires the structure of mixed-effects models to describe the biomarker's evolution over time with the specification of a link function.  In our approach, the longitudinal outcomes are assumed to follow a Poisson distribution, with the conditional mean modeled through a log-link function incorporating both fixed and random effects.
\begin{equation}
\label{eq:poisson_model}
\begin{aligned}
    (Y_{ijk} \mid \bm{b}_{ik}) &\sim \text{Poisson}(\lambda_{ijk}) \\
    \eta_{ijk}^{L} &= \log(\lambda_{ijk}) = \bm{x}_{ijk}^\top \bm{\beta}_k + \bm{z}_{ijk}^\top \bm{b}_{ik} \\
    (\bm{b}_{i} \mid {\bm\Sigma}) &\sim \mathcal{N}(\bm{0}, {\bm\Sigma}),
\end{aligned}
\end{equation}

\noindent where $\lambda_{ijk}$ is the conditional mean of the count result, $\bm{x}_{ijk}$ and $\bm{z}_{ijk}$  are covariate vectors for fixed and random effects, respectively;  $\bm{\beta}_k$ is a vector of fixed effects to the $k$-th biomarker; and $\bm{b}_{i} = \{\bm{b}_{ik};\ k = 1, \dots, K\}$ denote the vector of subject-specific random effects for individual $i$, assumed to follow a multivariate normal distribution with mean zero and variance-covariance matrix ${\bm\Sigma}$. 


For the survival sub-model, let $T_i$ be a non-negative random variable representing the true event time for the $i$-th individual, and let $C_i$ denote the censoring time. The observed survival time is then defined as $T^*_i = \min(T_i, C_i)$, and the corresponding censoring indicator is $\delta_i = \mathbb{I}(T_i \leq C_i)$. Thus, each individual is characterized by the pair $(T^*_i, \delta_i)$. The random effects from the longitudinal sub-model \eqref{eq:poisson_model} are incorporated into the hazard function to account for the shared dependence between the longitudinal and survival processes:
\begin{equation}
    \label{eq:join_model_proposed}
    h_{i}(t) = h_{0}(t) \, \exp(\eta_{i}^{S}) = h_{0}(t) \, \exp\left(\bm{w}_{i}^{\top} \bm{\psi} + \sum_{k=1}^{K}\bm{b}_{ik}^{\top} \bm{\gamma}_{k}\right), 
\end{equation}

\noindent where $h_{0}(t)$ is a baseline hazard function, $\bm{w}_{i}$ is a vector of baseline covariates quantifying the fixed effects with regression coefficients $\bm{\psi}$;  $\bm{b}_{ik}$  is the vector of subject-specific random effects corresponding to each of the $K$ biomarkers with association parameters $\bm{\gamma} = (\bm{\gamma}_1, \ldots, \bm{\gamma}_K)^\top$ linking the random effects to the survival outcome. This formulation jointly accounts for the longitudinal and survival processes across all biomarkers, for each individual $i = 1, \ldots, n$, and each biomarker $k = 1, \ldots, K$.

In many applications, the baseline hazard function is modeled nonparametrically to reduce the risk of misspecification. However, when exploratory data analyses suggest the presence of a cured subpopulation, it becomes important to explicitly account for this feature in the survival sub-model, especially when the goal is to understand how biomarkers influence the probability of cure. Ignoring the cure fraction may result in biased inference, particularly when the proportion of cured individuals is substantial and clinically relevant.  The class of defective distributions has been an interesting approach to model long-term survival data without the need of extra parameters in the statistical modeling, especially if the proportion of cured individuals matters. Defective distributions provide a natural and interpretable approach that allows for the presence (or absence) of a cure fraction to be accommodated within a unified framework, without requiring additional complexity in the estimation process.

Due to its great flexibility, we assume the Gompertz baseline hazard function in our joint model, i.e., $h_{0}(t) = \mu \,\exp(\alpha t)$. Considering a reparametrized version for the scale parameter $\mu = \exp(\gamma_0)$, for $\gamma_{0} \in \mathbb{R}$, the hazard function for the $i$-th individual is specified as
\begin{equation}
    \label{eq:hi(t)}
    h_{i}(t) = \exp(\alpha t + \gamma_{0} + \eta_{i}^{S}).
\end{equation}


The cumulative hazard function based on the hazard function in Equation \eqref{eq:hi(t)} is given by
\begin{align}
    H_{i}(t) &= \int_{0}^{t} h_{i}(u) \, du = \int_{0}^{t} \exp(\alpha u + \gamma_{0} + \eta_{i}^{S})  \, du =  \int_{0}^{t} \exp(\alpha u) \exp(\gamma_{0} + \eta_{i}^{S})  \, du \nonumber\\
            & = \exp(\gamma_{0} + \eta_{i}^{S}) \int_{0}^{t} \exp(\alpha u)   \, du = \exp(\gamma_{0} + \eta_{i}^{S}) \dfrac{1}{\alpha} \exp(\alpha u)\Big|_{0}^{t} =  \exp(\gamma_{0} + \eta_{i}^{S}) \dfrac{1}{\alpha} \left[\exp(\alpha t) - 1\right].
\end{align}

The cumulative hazard function plays a key role in deriving the survival function under the influence of random effects, through the relationship $S_{i}(t) = \exp\{-H_{i}(t)\}$. It is essential to assess how the random effects impact the individual survival probabilities and, consequently, the cure fraction for each individual. \\

\begin{theorem}
    \label{teo1}
    If $S_{0}(t)$ is defective, then $S_{i}(t)$ is also defective.
\end{theorem}

\begin{proof}
    Supose $\alpha < 0$. Then
    \begin{align}
    \label{theo1}
    \lim_{t \rightarrow \infty} S_{i}(t) &=  \lim_{t \rightarrow \infty} \exp\left\{- H_{i}(t)\right\} = \lim_{t \rightarrow \infty} \exp\left\{- \exp(\gamma_{0} + \eta_{i}^{S}) \dfrac{1}{\alpha} \left[\exp(\alpha t) - 1\right]\right\} \nonumber \\
    &= \exp\left\{- \exp(\gamma_{0} + \eta_{i}^{S}) \dfrac{1}{\alpha} \lim_{t \rightarrow \infty}  \left[\exp(\alpha t) - 1\right]\right\} =  \exp\left\{\exp(\gamma_{0} + \eta_{i}^{S}) \dfrac{1}{\alpha} \right\}.
\end{align}
\end{proof}

Since $\exp(\gamma_{0} + \eta_{i}^{S}) > 0$, it follows directly that Equation~\eqref{theo1} lies within the interval $(0,1)$. This completes the proof and demonstrates that the survival sub-model specified in \eqref{eq:join_model_proposed} can be interpreted as added random effects on the scale parameter. In other words, the latent intercepts and slopes contribute to the cure fraction of each individual through this parameter.

\subsection{Bayesian inference}

In the Bayesian paradigm for statistical inference, the use of prior distributions serves to express the knowledge (or ignorance) about the vector of parameters and its combination with likelihood function provides the posterior distribution, which is the mathematical object required to proceed with the inference \citep{murteira2018bayesiana}. Assuming $\bm{\theta}_{i} = (\bm{b}_{i}, \bm{\vartheta}) $, where $\bm{\vartheta} = (\bm{\beta}, \bm{\gamma}_{k}, \Sigma_{k}, \alpha, \gamma_{0})$ is the vector of fixed effects, $\bm{b}_{i} = (\bm{b}_{i1}, \dots, \bm{b}_{iK})$ is the vector of random effects and $D_{i} = \{T_{i}^{*}, \delta_{i}, Y_{ijk}; j = 1, \ldots, n_{i}; k = 1, \ldots, K\}$ is the data with the observed variables, $i=1,\dots,n$. In this way, the full vector of parameters is $\bm{\theta} = \left\{\bm{\theta}_{i}, i = 1, \ldots, n\right\}$ and the complete dataset is $\bm{D} =  \left\{D_{i}, i = 1, \ldots, n\right\}$. The posterior distribution is given by the Bayes theorem
$$\pi(\bm{\theta} \mid \bm{D}) = \dfrac{\pi (\bm{D} \mid \bm{\theta}) \,\pi(\bm{\bm{\theta}}) }{\pi (\bm{D})} \varpropto \pi (\bm{D} \mid \bm{\theta}) \,\pi(\bm{\theta}),$$

\noindent where $\pi (\bm{D} \mid \bm{\theta})$ is the likelihood function, $\pi(\bm{\theta})$ is the prior distribution and $\pi (\bm{D})$ is the distribution of data unconditioned on parameters. The quantity $\pi (\bm{D}) = \int_{\bm{\theta}} \pi (\bm{D} \mid \bm{\theta}) \,\pi(\bm{\theta}) \, d\bm{\theta}$, also known as a normalizing constant, is analytically untractable in most cases. Bayesian inference traditionally relies on sampling-based algorithms, such as Markov chain Monte Carlo (MCMC), to estimate summary statistics from the posterior distribution, marginal densities and predictive distributions. However, approximate methods such as the integrated nested Laplace approximation (INLA) offer a computationally efficient alternative by providing accurate approximations of posterior densities at substantially lower computational cost compared to MCMC.


The full likelihood function is derived in terms of the advantage of conditional independence between the biomarker $Y_{ijk}$ and the event time $T^{*}_{i}$,
\begin{equation}
    \pi(D_{i} \mid \bm{\theta}_{i}) = \int_{\bm{b}_{i}}  
    \left[
        \pi(T_{i}^{*}, \delta_{i} \mid \bm{\theta}_{i}) 
        \, \displaystyle\prod_{k=1}^{K} \displaystyle\prod_{j=1}^{n_{i}} 
        \pi(Y_{ijk} \mid \bm{\theta}_{i}) \, 
        \pi(\bm{b}_{i})
    \right] 
    \, d \bm{b}_{i},
\end{equation}

\noindent where $\pi(T_{i}^{*}, \delta_{i} \mid \bm{\theta}_{i}) = h_{i}(T_{i}^{*} \mid \bm{\theta}_{i})^{\delta_{i}} \, \exp \left(- \int_{0}^{T_{i}^{*}} h_{i}(T_{i}^{*} \mid \bm{\theta}_{i})\right)$ corresponds to the likelihood contribution from the survival component, based on the hazard function presented in Equation \eqref{eq:hi(t)}; $\pi(Y_{ijk} \mid \bm{\theta}_{i}) = \dfrac{\lambda_{ijk}^{Y_{ijk}} \, \exp(-\lambda_{ijk})}{Y_{ijk} !}$ is the probability mass function for the $k$-th longitudinal sub-model; and $ \pi(\bm{b}_{i})$ is the density of the random effects.


\subsection{Integrated nested Laplace approximation}

Joint models incorporating Gaussian random effects can be formulated within the framework of latent Gaussian models, thereby enabling statistical inference through the integrated nested Laplace approximation (INLA) method \citep{alvares2024, rustand2024joint}. The joint model needs to be expressed with a hierarchical structure with three layers: ($i$) The likelihood function of the observed data is assumed to be conditionally independent; ($ii$) the latent field $\chi$ defined by a multivariate Gaussian distribution conditioned on a set of hyperparameters and; ($iii$) a prior distribution elicited with a set of hyperparameters.

The first step is to formulate the defective joint model as a latent Gaussian model. We decompose the full parameter vector ($\bm{\theta}$) into a Gaussian latent field $\chi = (\bm{\eta}^{L}, \bm{\eta}^{S}, \bm{b}, \bm{\beta}, \bm{\psi}, \bm{\gamma}, \gamma_{0})$ and the set of hyperparameters $\bm{\varphi}$ containing the likelihood-specific
parameters, such as $\alpha$ and variance and covariance of random effects; where $\bm{\eta}^{L}= \{\eta_{ijk}^{L}, i = 1, \dots, n; j = 1, \dots, n_{i}; k = 1, \dots, K\}$ are the linear predictors for the longitudinal Poisson outcomes $Y_{ijk}$; $\bm{\eta}^{S} = \{\eta_{i}^{S}, i = 1, \dots, n\}$ are the linear predictors for the Gompertz baseline hazard; $\bm{b}$ are the shared subject-specific random effects across the longitudinal and survival models; $\bm{\beta}$, $\bm{\psi}$ and $\bm{\gamma}$ are fixed effect coefficients for the longitudinal, survival and association structures, respectively; $\alpha$ and $\gamma_{0}$ are the parameters of the Gompertz baseline hazard, i.e., $h_{0}(t) = \exp(\alpha t + \gamma_{0})$.

We assume the random effects $\bm{b}_{i} \sim \mathcal{N}(\bm{0}, \bm{Q}^{-1}_{b})$, where $\bm{Q}_{b}$ is the precision matrix associated with the shared Gaussian random effects, containing the precision and correlation hyperparameters $\bm{\tau}_{\bm{b}}$. Additionally, 
\begin{align*}
    \bm{\beta} &\sim \mathcal{N}(\bm{0}, \bm{\tau}_{\bm{\bm{\beta}}}^{-1} I), \\
    \bm{\psi} &\sim \mathcal{N}(\bm{0}, \bm{\tau}_{\bm{\psi}}^{-1} I), \\
    \bm{\gamma} &\sim \mathcal{N}(\bm{0}, \bm{\tau}_{\bm{\gamma}}^{-1} I), \\
    \gamma_{0} &\sim \text{Normal}(0, 1/\tau_{\gamma_{0}}), \\
    \alpha &\sim \text{Normal}(0, 1/\tau_{\alpha}).
\end{align*}

Then, the latent field $\chi$ forms a Gaussian markov random field with sparse precision matrix $\bm{Q}(\bm{\varphi}_{1})$, i.e., $$(\chi \mid \bm{\varphi}) \sim \mathcal{N}(\bm{0}, \bm{Q}^{-1}(\bm{\varphi})).$$

In the Bayesian formulation, a prior can be defined fot the set $\bm{\varphi}$, $\pi(\bm{\varphi})$. We assume the distribution of the observation variables $\bm{D} = \{D_{i}; i =1, \dots, n\}$ is conditionally independent given the Gaussian random field $\chi$ and the hyperparameters $\bm{\varphi}$, i.e., $$(\bm{D} \mid \chi, \bm{\varphi}) \sim \prod_{i=1}^{n} \pi(D_{i} \mid \chi_{i}, \bm{\varphi}).$$

In this way, the joint posterior distribution of $(\chi, \bm{\varphi})$ can be written as
\begin{align*}
    \pi(\chi, \bm{\varphi} \mid \bm{D}) & \varpropto \pi(\bm{\varphi}) \,  \pi(\chi \mid \bm{\varphi})  \, \prod_{i=1}^{n} \pi(D_{i} \mid \chi_{i}, \bm{\varphi}) \\
        &\varpropto  \pi(\bm{\varphi}) \, | \bm{Q}(\bm{\varphi})|^{n/2} \, \exp \left\{\frac{1}{2} \, \chi^{\top} \bm{Q}(\bm{\varphi}) \, \chi + \sum_{i=1}^{n} \log\left[\pi(D_{i} \mid \chi_{i}, \bm{\varphi})  \right]\right\}.
\end{align*}


The main goal of INLA is to use subsequent approximations for marginal posterior densities $\pi(\chi_{i} \mid \bm{D}), i = 1, \dots, n$, and the joint posterior density $\pi(\chi, \bm{\varphi} \mid \bm{D})$. 


We proceed with a brief explanation of how the INLA handles the approximation task into three steps:

\begin{enumerate}
    \item Approximate the marginal posterior density distribution of the hyperparameters using the Laplace approximation:

    $$
    \pi(\bm{\varphi} \mid  \bm{D}) = \frac{\pi(\chi, \bm{\varphi} \mid  \bm{D})}{\pi(\chi \mid \bm{\varphi} ,  \bm{D})} \approx \frac{\pi(\bm{\varphi}) \pi(\chi \mid \bm{\varphi}) \pi(\chi \mid  \bm{D}, \bm{\varphi})}{\tilde{\pi}(\chi \mid \bm{\varphi} ,  \bm{D})}  \Big{|}_{\chi = \chi^{*}(\bm{\varphi})}.  
    $$

\noindent In the above step, the goal is to propose an approximation to the denominator based on the mode $\chi^{*}(\bm{\varphi})$ of the latent field for a given configuration of $\bm{\varphi}$ using the Laplace approximation.

    \item Approximate the conditional posterior distribution of the latent field 

    $$
    \pi(\chi_{j} \mid \bm{\varphi},  \bm{D}) \varpropto \dfrac{\pi( \chi, \bm{\varphi} \mid  \bm{D})}{\pi(\chi_{-j} \mid \chi_{j}, \bm{\varphi},  \bm{D})}.
    $$

    There are three options to conduct this approximation: (a) using a Gaussian approximation \citep{rue2009approximate}, (b) propose again a Laplace approximation \citep{rue2009approximate} or (c) by expanding the numerator and denominator up to a third order Taylor expansion and then apply the Laplace  approximation.

    \item Approximate the marginal posterior distributions of the latent field using numerical integration,

    $$
    \pi(\chi_{j} \mid  \bm{D}) \approx \sum_{h=1}^{H} \tilde{\pi}(\chi_{j} \mid \bm{\varphi}_{h}^{*},  \bm{D}) \, \tilde{\pi}( c \mid \bm{D}) \, \Delta h,
    $$

    from steps 1 and 2. The integration points $\left\{ \bm{\varphi}_{1}, \dots, \bm{\varphi}_{H} \right\}$ are selected through a rotation using polar coordinates, taking into account the density distribution of these points. The quantity $\Delta h$ are the corresponding weights. 
\end{enumerate}

In the explanation above, the approximation of the posterior marginal distributions for each element of the latent field and for each hyperparameter, achieved through numerical integration, constitutes the ``integrated'' component of the INLA algorithm, whereas the first two steps correspond to the ``nested Laplace'' approximation component. Recent advances in INLA have introduced a low-rank variational Bayes correction that aligns the approximation with the full Laplace method with no additional computational cost compared to the first Gaussian approximation \citep{van2023new}. This new formulation is available in \texttt{INLAjoint} R-package \citep{rustand2024fast, alvares2024, rustand2024joint}.

\section{Simulation studies} \label{sec:simul}

In our simulation studies, we investigated three sample sizes ($n= 100, 500, 1000$), each subject has four longitudinal records (one at baseline and thereafter three visits), they are scheduled at equally spaced time interval $t_{ij} = (0, 0.3, 0.6, 0.9)$. After this period, subjects are said to be censored non-informatively. The true conditional mean to a longitudinal profile follows the following structure:
\begin{equation}
    \label{eq:sim_long_submodel}
    \lambda_{i}(t_{ij}) = \beta_{0i} + b_{0i} + (\beta_{1} + b_{1i}) \, t_{ij} + \beta_2 x_{1i} +  \beta_3 x_{2i} ,
\end{equation}

\noindent where $x_{1i} \sim \text{Bernoulli}(0.8)$ and $x_{2i} \sim \text{Uniform}(0,1)$ are fixed covariates with $\beta_2$ and $\beta_3$ as coefficients; $\beta_{0}$ is the fixed intercept and $\beta_{1}$ is the fixed time-effect; the random effects $(b_{0i}, b_{1i})$ are generated from a mean-zero bivariate Normal distribution with variances $\sigma_{b_{0}}^2$ and $\sigma_{b_{1}}^2$, respectively, and correlation $\rho$. The longitudinal responses are generated based on the conditional mean, i.e., $Y_{ij} \sim \text{Poisson}(\lambda_{i}(t_{ij}))$. For clarity and simplicity, our analysis focuses on a single longitudinal biomarker. Nonetheless, extending the model to handle multiple biomarkers is straightforward, as illustrated in the application in Section~\ref{sec:application}.
 
For the survival sub-model, we define the conditional hazard function as
\begin{equation}
    \label{eq:sim_surv_submodel}
    h_{i}(t \mid x_{1i}, b_{0i}, b_{1i}) = \exp(\alpha t  + \gamma_{0} + x_{1i}  \psi_{1} +  b_{0i} \gamma_{01} + b_{1i} \gamma_{11}),
\end{equation}

\noindent where we share the individual random effects (intercept and slope).

Our aim is to simulate failure times ($T_{i}^{*}$, $i = 1, \dots,n$) conditioned on the cure probability and the latent effects from the longitudinal process, as captured by the conditional hazard function. To achieve this, we proceed with the following steps:

\begin{enumerate}
    \item Generate the value of cure fraction  obtained from Theorem \ref{teo1}
    
    $$p_{i} = \exp\left\{ \dfrac{1}{\alpha} \, \exp(\gamma_{0} + x_{1i}  \psi_{1} +  b_{0i} \gamma_{01} + b_{1i} \gamma_{11}) \right\};$$
    
    \item Generate $M_{i} \sim \text{Bernoulli}(1 - p_{i})$;
    \item If $M_{i}=0$ set $t_{i} = \infty$. If $M_{i}=1$ take $t_{i}^{'}$ as the root of $F(t) = u$ given by 

    $$
    t_{i} = \frac{1}{\alpha} \, \log\left(1 - \frac{\alpha \, \log(1-u)}{\exp(\gamma_{0} +  x_{1i}  \psi_{1} +  b_{0i} \gamma_{01} + b_{1i} \gamma_{11})}  \right),
    $$

    \noindent where $u \sim \text{Uniform}(0,1-p_{i})$;

    \item Generate $u^{*}_{i} \sim \text{Uniform}(0, \max(t_i))$, considering only the finite $t_{i}$;

    \item Calculate $t_i^{*} = \min(t_{i}, u^{*}_{i})$ and define the event indicator $\delta_{i} = 1$ if $t_{i} < u^{*}_{i}$, and $\delta_{i} = 0$ otherwise.
\end{enumerate}

We conducted two Monte Carlo experiments under vague prior specification. The first scenario (Table \ref{tab:experiment_sim1}) was designed to resemble the values observed for one of the biomarkers analyzed in the application (see Section~\ref{sec:application}). In this case, the samples exhibit a high percentage of censoring, which corresponds to a greater presence of a cure fraction. The second scenario (Table \ref{tab:experiment_sim2}) explores the asymptotic behavior of the Bayesian model under conditions of lower censoring and a smaller proportion of long-term survivors. 

The code developed for the simulation study and the application of the joint cure model was written in R, using the \texttt{INLAjoint} package \citep{alvares2024,rustand2024joint} and the likelihood function of the defective Gompertz distribution (\texttt{dgompertz}) available in \texttt{R-INLA} package. All codes containing simulation studies are available on \url{https://github.com/Dionisioneto/Code-for-papers/tree/main/new\_joint\_cure\_model\_dgompertz}.

The simulation studies aim to investigate the behavior of bias and coverage probability of the joint model under vague prior specifications across 1000 replicates. The frequentist properties were evaluated in both scenarios. Overall, the frequentist properties are evaluated as satisfactory across both scenarios. The credible intervals exhibit appropriate coverage probabilities, consistently capturing the true parameter values at the expected rates. As the sample size increases, the bias declines reasonably for all components of the parameter vector. The results presented in Tables \ref{tab:experiment_sim1} and \ref{tab:experiment_sim2} indicate that the proposed model performs well under both high and low censoring and cure rate scenarios.

\begin{table}[h]
\centering
\caption{The average bias and coverage probability (C.P.) for the 95\% credible intervals of 1000 Monte Carlo replicates of the first scenario of censoring (50–70\%).}
\small
\begin{tabular}{lcccccccccccc} \\ \hline
Parameter & $\alpha$ & $\gamma_{0}$ & $\gamma_{1}$ & $\gamma_{3}$ & $\psi_{1}$ & $\beta_{0}$ & $\beta_{1}$ & $\beta_{2}$ & $\beta_{3}$ & $\sigma_{b_{0}}$ & $\sigma_{b_{1}}$ & $\rho$ \\ \hline
True      & -0.65     & -0.68        & 0.68         & 0.17         & -0.37      & 2.50        & -0.20       & -0.01       & 0.10        & 0.25             & 0.25             & -0.05                  \\ \hline
\multicolumn{13}{c}{$n=100$} \\ \hline
Bias      & -0.032    & -0.070       & -0.296       & -0.040       & 0.030      & -0.003      & -0.002      & -0.001      & 0.001       & 0.021            & 0.047            & -0.068                 \\
C.P. (95\%)& 0.939     & 0.920        & 0.973        & 0.989        & 0.943      & 0.950       & 0.961       & 0.961       & 0.949       & 0.927            & 0.921            & 0.973                  \\ \hline
\multicolumn{13}{c}{$n=500$} \\ \hline
Bias      & 0.000     & -0.019       & -0.084       & 0.021        & 0.005      & -0.002      & -0.001      & 0.000       & 0.001       & 0.006            & 0.016            & -0.034                 \\
C.P. (95\%)& 0.954     & 0.928        & 0.961        & 0.963        & 0.956      & 0.943       & 0.944       & 0.944       & 0.943       & 0.930            & 0.957            & 0.943                  \\ \hline
\multicolumn{13}{c}{$n=1000$} \\ \hline
Bias      & -0.002    & -0.003       & -0.032       & 0.002        & 0.000      & 0.000       & 0.000       & 0.000       & 0.000       & 0.003            & 0.009            & -0.020                 \\
C.P. (95\%)& 0.949     & 0.921        & 0.954        & 0.962        & 0.939      & 0.950       & 0.945       & 0.956       & 0.953       & 0.949            & 0.943            & 0.945                  \\ \hline
\end{tabular}
\label{tab:experiment_sim1}
\end{table}

\begin{table}[h]
\centering
\caption{The average bias and coverage probability (C.P.) for the 95\% credible intervals of 1000 Monte Carlo replicates of the second scenario of censoring (10–30\%).}
\small
\begin{tabular}{lcccccccccccc} \\ \hline
Parameter & $\alpha$ & $\gamma_{0}$ & $\gamma_{1}$ & $\gamma_{3}$ & $\psi_{1}$ & $\beta_{0}$ & $\beta_{1}$ & $\beta_{2}$ & $\beta_{3}$ & $\sigma_{b_{0}}$ & $\sigma_{b_{1}}$ & $\rho$ \\ \hline
True      & -1.00     & 0.80         & 1.00         & 1.00         & 0.50       & 1.00        & -1.00       & -0.10       & 0.50        & 0.50             & 0.50             & 0.40                   \\ \hline
\multicolumn{13}{c}{$n=100$} \\ \hline
Bias       & -0.006    & -0.021       & 0.068        & -0.249       & 0.017      & -0.004      & 0.003       & -0.003      & 0.007       & -0.009           & 0.002            & -0.060                  \\
C.P. (95\%) & 0.951     & 0.927        & 0.982        & 0.969        & 0.951      & 0.945       & 0.934       & 0.942       & 0.951       & 0.962            & 0.996            & 0.985                   \\ \hline
\multicolumn{13}{c}{$n=500$} \\ \hline
Bias       & -0.011    & -0.006       & 0.036        & -0.106       & 0.005      & -0.002      & -0.001      & 0.005       & 0.000       & -0.004           & -0.004           & -0.002                  \\
C.P. (95\%) & 0.963     & 0.935        & 0.965        & 0.948        & 0.951      & 0.948       & 0.926       & 0.937       & 0.956       & 0.949            & 0.982            & 0.972                   \\ \hline
\multicolumn{13}{c}{$n=1000$} \\ \hline
Bias       & -0.012    & -0.003       & 0.024        & -0.069       & 0.006      & -0.002      & 0.000       & -0.001      & 0.001       & -0.004           & -0.005           & 0.007                   \\
C.P. (95\%) & 0.948     & 0.927        & 0.964        & 0.965        & 0.939      & 0.958       & 0.912       & 0.945       & 0.956       & 0.955            & 0.973            & 0.972                   \\ \hline
\end{tabular}
\label{tab:experiment_sim2}
\end{table}

\section{The Standard and New Antiepileptic Drugs (SANAD) study} \label{sec:application}

Epilepsy is a long-term neurological condition marked by repeated and unpredictable seizures, caused by abrupt and intense bursts of electrical activity in clusters of neurons across different brain regions. Identifying the disorder for clinical diagnostic involves assessing variables such as the location where seizures originate, the size of the brain area involved, and how often the seizures occur \citep{berto2002quality}. However, with advances in medical science, most individuals with epilepsy are able to lead stable lives without experiencing cognitive or psychiatric disorders, primarily through the use of medication or surgical interventions. A smaller proportion of patients, however, may develop difficulties with attention and memory in daily tasks, and in some cases, even psychotic symptoms.

The antiepileptic drugs can be divided into two categorias: the classical, such as carbamazepine (CBZ), phenobarbital (PHB), and phenytoin (PHT); and the modern ones such as gabapentin (GBP), levetiracetam (LEV), lacosamide (LCM), lamotrigine (LTG), pregabalin (PGN), topiramate (TPM), and zonisamide (ZNS), which are less likely to be involved in drug interactions based on enzyme induction \citep{berto2002quality}. The antiepileptic drugs represent the first-line treatment for seizure control. Evidence suggests that these medications have a positive impact on quality of life of patients by reducing the probability of recurrence of epileptic seizure \citep{fitton1995lamotrigine, birbeck2002seizure, shorvon2015drug}. 

Given the high proportion of patients who achieve a good quality of life with antiepileptic medication, it becomes reasonable to model the time to treatment failure (i.e., the occurrence of an epileptic seizure) within a long-term survivor framework. This approach acknowledges the existence of a subgroup of patients who are statistically seizure-free throughout pharmacological treatment-some of whom may eventually discontinue medication altogether due to their sustained low risk of seizure recurrence \citep{schmidt2009drug,saetre2010antiepileptic}.


Motivated by the prolonged time to treatment failure commonly observed in patients undergoing long-term antiepileptic drug therapy, we apply a joint modeling approach to data from a randomized controlled trial comparing standard and newer antiepileptic drugs in terms of their clinical effects and impact on quality of life in epilepsy treatment. Using the dataset available in the \texttt{joineRML} R-package \citep{hickey2018joinerml}, we accessed a sample from the SANAD (Standard and New Antiepileptic Drugs) study, as reported in \cite{marson2007randomised} and \cite{marson2007sanad}. The original aim of the study was to assess quality of life measures --including anxiety, depression, and clinical profiles related to adverse events-- using discrete scores obtained from questionnaires administered at four time points: baseline, 3 months, 1 year, and 2 years after the initial visit. In the observed sample, the number of longitudinal records per patient is distributed as follows: 5.51\% have one record, 11.95\% have two records, 19.12\% have three records, and 63.42\% have four records.

The dataset includes both longitudinal and survival data for 544 patients who were randomly assigned to receive one of two types of antiepileptic drugs: the standard treatment carbamazepine (CBZ) or the newer alternative lamotrigine (LTG). Accordingly, the primary focus of the study is to evaluate the time until treatment failure --defined as the occurrence of a seizure-- for each treatment group. Among the patients, we have 41.54\% of censoring. The Kaplan-Meier curve presented in Figure \ref{fig:surv_curve_cf_est} for each treatment group reveals the presence of a plateau that indicates the existence of long-term survivors.

\begin{figure}[h!]
    \centering
    \includegraphics[width=0.5\linewidth]{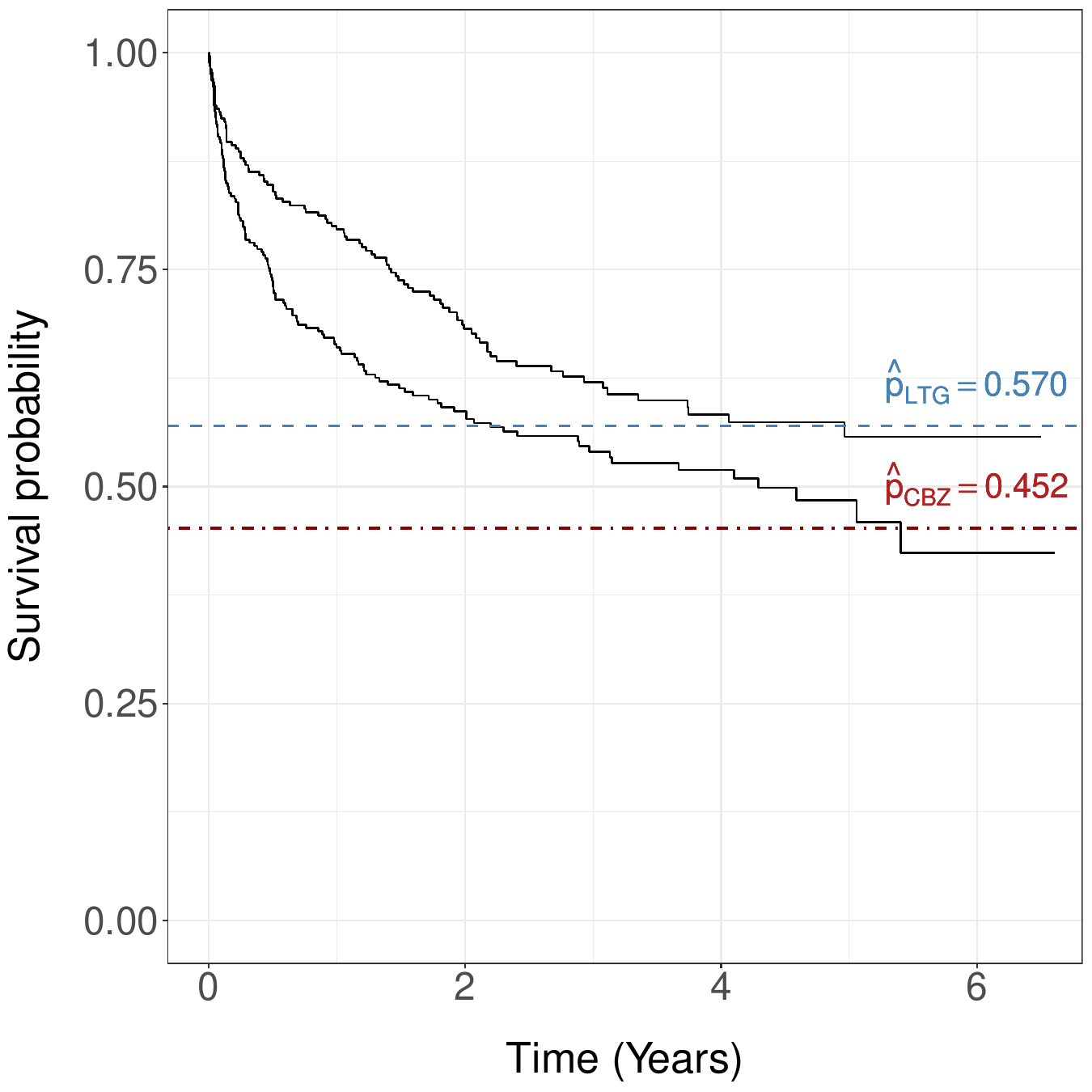}
    \caption{Survival curve of non-parametric Kaplan-Meier estimator with estimated cure fraction of both placebo (red) and new drug (blue) groups.}
    \label{fig:surv_curve_cf_est}
\end{figure}

It is important to note that this study has previously been reported in \cite{jacoby2015quality} using a larger sample. In that work, the authors applied both a linear mixed-effects model and a joint model to investigate the association between quality of life (QOL) measurements. However, the linear mixed-effects model exhibited convergence issues, likely due to the assumption of normality applied to biomarkers that are, in fact, discrete. This mismatch may compromise the inferential validity of the joint model, especially when incorporating multiple QOL outcomes. Therefore, in the present study, we model the two count-based outcomes using the Poisson distribution, as the mean and variance of the biomarkers are relatively similar.

We applied our joint model based on the defective Gompertz distribution to detect whether the shape parameter ($\alpha$) identifies the presence of a cure fraction in the data. The following conditional mean structures were considered for the two QOL biomarkers:
\begin{align*}
    \log(\mathbb{E}[\text{anxiety}_{ij}]) &= \beta_{01} + b_{01i} + (\beta_{11} + b_{11i}) t_{ij} + \beta_{21} \text{drug}_{i}, \\
    \log(\mathbb{E}[\text{depress}_{ij}]) &= \beta_{02} + b_{02i} + (\beta_{12} + b_{12i}) t_{ij} + \beta_{22} \text{drug}_{i}, \\
    (\bm{b}_{i} \mid {\bm\Sigma}) & \sim \mathcal{N}(\bm{0}, {\bm\Sigma}), \\
    {\bm\Sigma} &= 
    \begin{bmatrix}
    \sigma^2_{b_{01}} &  &  &  \\
    \rho_{(b_{01}, b_{11})} \sigma_{b_{01}} \sigma_{b_{11}} & \sigma^2_{b_{11}} &  &  \\
    \rho_{(b_{01}, b_{02})} \sigma_{b_{01}} \sigma_{b_{02}} & \rho_{(b_{11}, b_{02})} \sigma_{b_{11}} \sigma_{b_{02}} & \sigma^2_{b_{02}} & \\
    \rho_{(b_{01}, b_{12})} \sigma_{b_{01}} \sigma_{b_{12}} & \rho_{(b_{11}, b_{12})} \sigma_{b_{11}} \sigma_{b_{12}} & \rho_{(b_{02}, b_{12})} \sigma_{b_{02}} \sigma_{b_{12}} & \sigma^2_{b_{12}}
    \end{bmatrix},
\end{align*}

\noindent and a time-to-event sub-model for the study of failure of antiepileptic treatment:
$$
h_{i}(t) = \exp(\alpha t + \gamma_{0} + \psi_{1} \text{drug}_{i} + \gamma_{01} b_{01i} + \gamma_{11} b_{11i} + \gamma_{02} b_{02i} + \gamma_{12} b_{12i}).
$$

Table~\ref{tab:inferential_sumarry} presents the inferential summary of the Bayesian defective joint model applied to the epilepsy dataset. The first point of analysis concerns the parameter $\alpha < 0$, indicating the presence of a cure fraction under the integrated approach with the defective distribution. Together with the 95\% credible interval, this estimate provides evidence that the experimental data include patients who are not considered at risk of epileptic seizures.

\begin{table}[h!]
\centering
\caption{Results of posterior mean, posterior standard deviation (SD) and 95\% credible interval (CI) for the application of epileptic dataset with 2 QOL biomarkers (anxiety and depress).}
\begin{tabular}{lccc} \\ \hline
\multicolumn{1}{c}{Parameter} & Mean   & SD    & 95\% CI            \\ \hline
$\alpha$                      & -0.659 & 0.074 & (-0.807; -0.516) \\
$\gamma_{0}$                  & -0.680 & 0.092 & (-0.860; -0.500)    \\
$\gamma_{01}$                 & 0.919  & 0.571 & (-0.209; 2.039)  \\
$\gamma_{11}$                 & 0.916  & 0.898 & (-0.822; 2.714)  \\
$\gamma_{02}$                 & 0.680  & 0.632 & (-0.558; 1.931)  \\
$\gamma_{12}$                 & 0.178  & 0.893 & (-1.602; 1.914)  \\
$\psi_{1}$                    & -0.370 & 0.139 & (-0.643; -0.096)   \\
$\beta_{01}$                  & 2.680  & 0.021 & (2.640; 2.720)     \\
$\beta_{11}$                  & -0.022 & 0.010 & (-0.041; -0.002) \\
$\beta_{21}$                  & 0.034  & 0.028 & (-0.021; 0.089)  \\
$\beta_{02}$                  & 2.560  & 0.020 & (2.522; 2.599)   \\
$\beta_{12}$                  & -0.027 & 0.011 & (-0.048; -0.006) \\
$\beta_{22}$                  & -0.012 & 0.027 & (-0.064; 0.041)  \\
$\sigma_{b_{01}}$             & 0.287  & 0.013 & (0.264; 0.313)   \\
$\sigma_{b_{11}}$             & 0.111  & 0.007 & (0.098; 0.125)   \\
$\sigma_{b_{02}}$             & 0.260  & 0.012 & (0.237; 0.285)   \\
$\sigma_{b_{12}}$             & 0.112  & 0.007 & (0.100; 0.125)   \\
$\rho_{b_{01}, b_{11}}$       & -0.067 & 0.087 & (-0.236; 0.106)  \\
$\rho_{b_{01}, b_{02}}$       & 0.805  & 0.025 & (0.751; 0.849)   \\
$\rho_{b_{01}, b_{12}}$       & 0.043  & 0.081 & (-0.112; 0.198)   \\
$\rho_{b_{11}, b_{02}}$       & 0.026  & 0.090 & (-0.150; 0.202)   \\
$\rho_{b_{11}, b_{12}}$       & 0.180  & 0.067 & (0.044; 0.305)   \\
$\rho_{b_{02}, b_{12}}$       & -0.059 & 0.082 & (-0.217; 0.100)     \\ \hline
\label{tab:inferential_sumarry}
\end{tabular}
\end{table}


The association of the random effects for the first biomarker (anxiety) showed similar magnitudes on the risk of seizure occurrence. Although not significant, these positive effects suggest that a higher score of anxiety compared to population average at baseline, as well as a steeper increase over time may be associated with a higher risk of seizures. Similarly, for the second biomarker (depression), the random intercept had a greater average impact than the random slope of time on the risk of treatment failure, although none of the coefficients was statistically relevant at the 95\% level. The treatment effect from antiepileptic drugs revealed an expected reduction of $\exp(-0.37) = 0.69$, which means that the risk is reduced on average by 31\% using the modern drug LTG compared to traditional CBZ. To quantify the uncertainty regarding the effectiveness of the newer drug in preventing epileptic seizures, we used 1000 posterior samples of the fixed effect parameter $\psi_{1}$, and exponentiated these values as previously described. The 2.5\% and 97.5\% quantiles computed from the posterior sample of the estimator were 0.52 and 0.92, respectively. These values indicate that the reduction in seizure risk associated with the LTG drug, relative to the older CBZ drug, ranges from 48\% to 8\%.

Regarding fixed effects related to biomarkers, both intercept and slope were important for both outcomes, with a slight decrease observed over time. For both biomarkers, the effect of the drug (LTG) was small or negligible. Concerning the random effects parameters, the estimated variability for both the intercepts and slopes was similar across biomarkers, even at the extremes of the 95\% credible intervals. The highest correlation was observed between the random intercepts, as also reflected in the 95\% credible interval, which means that a higher score of anxiety compared to population average at baseline is associated to a higher score for depression. However, the near-zero correlation between the slopes suggests independent evolutions over time.

In addition to the description of survival curves by treatment group, where it is evident that patients treated with LTG exhibit a higher probability of survival compared to those treated with the traditional drug CBZ. Figure \ref{fig:surv_curve_cf_est} also illustrates the presence of proportional hazards between the groups. Furthermore, Figure \ref{fig:surv_curve_cf_est} highlights one of the main advantages of using models that account for a cure fraction: the estimation of the cure proportion within each treatment group. To obtain these estimates from the parametric model, we used 1000 posterior samples from the joint model, incorporating the random effects. For each sample point, we computed the cure fraction for each individual based on the result presented in Theorem \ref{teo1}. The expected values for each individual within each treatment group, along with their corresponding 95\% credible intervals, are shown in Figure \ref{fig:cure-prob-groups} in Appendix~\ref{appendix1}. Based on the posterior mean of the cure probability per individual, we used the mean value to estimate the plateau of the cure fraction in each treatment group. The average cure fraction for the CBZ group was 0.453, while for the LTG one it was 0.570.

Similarly, using the upper and lower bounds of the 95\% credible intervals for each patient, we calculated the mean limits to reflect the uncertainty associated with the probability of being cured within each group. For patients treated with the modern drug LTG, the average cure fraction ranged between (0.314; 0.755). In contrast, for the group receiving doses of the traditional drug CBZ, the cure fraction was between (0.213; 0.666). These results suggest that LTG treatment provides patients with a greater chance of living seizure-free compared to conventional CBZ treatment, with the influence of longitudinal associations of anxiety and depression scores. Several previously mentioned studies in the medical and pharmaceutical fields support the greater effectiveness of this treatment, as demonstrated in the present analysis.


We also explored the relevance of the defective form of the Gompertz distribution in the estimation process. Figure \ref{fig:surv_curve_gompertz_dgompertz_est} presents the survival curves derived from Equation \eqref{eq:surv_Gompertz}, using the mean posterior of the parameters $\alpha$ and $\gamma_0$ under both the defective Gompertz distribution and its conventional form (i.e., without expanding the parameter space to include negative real values). 

\begin{figure}[h!]
    \centering
    \includegraphics[width=0.5\linewidth]{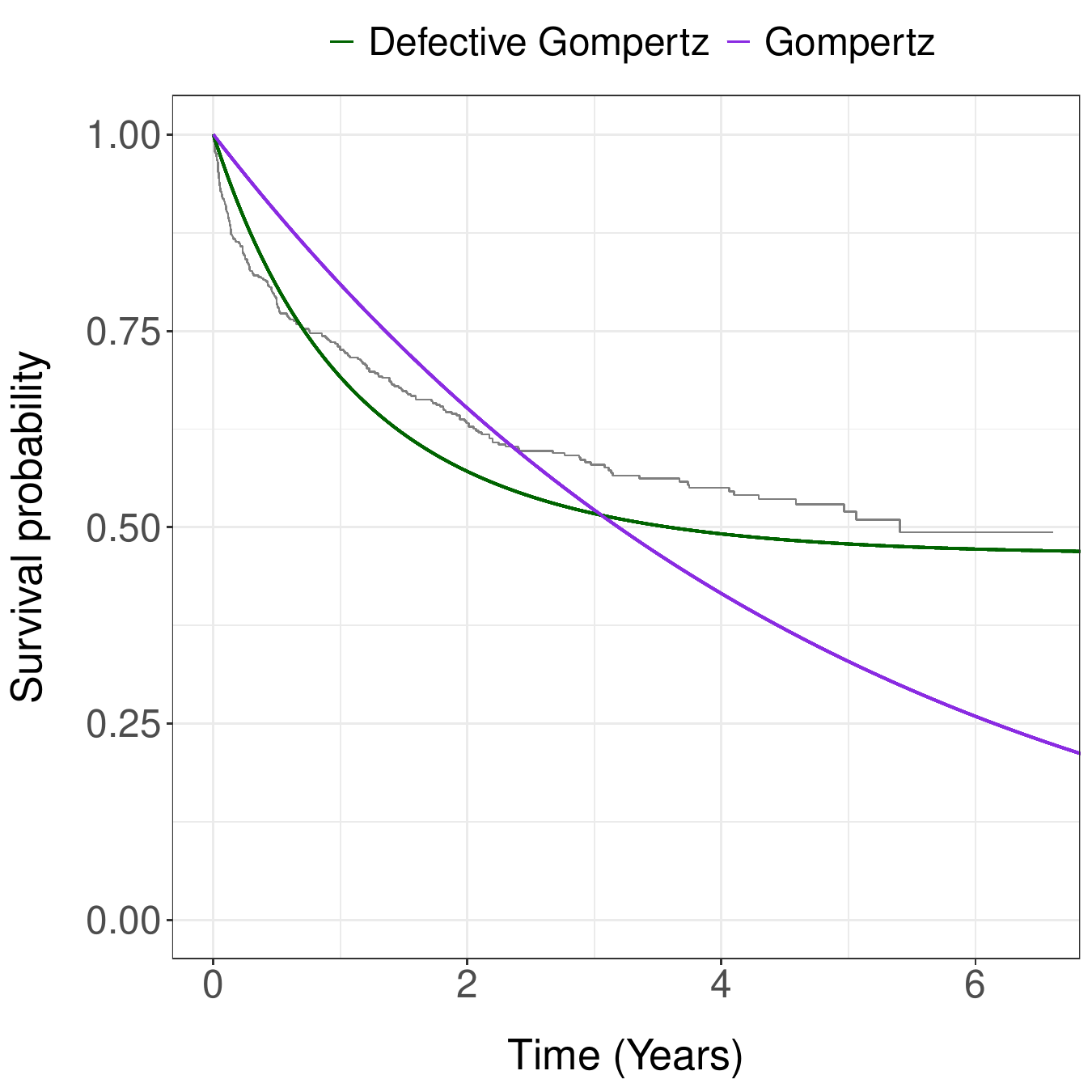}
    \caption{Comparison of the Kaplan-Meier survival curve and estimated curves from the defective Gompertz and Gompertz models.}
    \label{fig:surv_curve_gompertz_dgompertz_est}
\end{figure}

Unlike the defective version, the conventional form yielded approximate posterior means of $\alpha = 0.02$ and $\gamma_0 = -1.57$. These values directly influence the shape of the survival curve used to represent censored data. By not accounting for the possibility of a cure fraction, the joint model based on the conventional Gompertz distribution as the baseline produces a survival curve that monotonically decreases towards zero as follow-up time increases, thereby failing to capture the plateau observed in the nonparametric Kaplan–Meier estimate. This limitation highlights the advantages of adopting the more general defective formulation, which remains applicable even in the absence of a cured subpopulation in the observed sample.

\section{Conclusions} \label{sec:concl}

In this study, we have proposed a joint model for longitudinal and survival data that accounts for the possibility that individuals are immune to the event of interest. Using the defective Gompertz distribution, our proposal has integrated the presence of long-term survivors and allowed for the estimation of the cure fraction without requiring additional structural assumptions. This class of models offers a robust alternative to conventional mixture cure models by mitigating potential convergence issues commonly encountered with them. This work advances the application of defective distributions in joint modeling of long-term survivors, introducing new directions for inverse Gaussian and Dagum distributions, which have gained attention in the literature. In these models, the presence or absence of a cure fraction is also determined directly through a single parameter of the distribution, allowing for an integrated and efficient approach.

Our approach relies on Bayesian inference for parameter estimation, employing the fast and accurate INLA method. Simulation studies have demonstrated that the proposed estimation method yields estimators with low bias and satisfactory coverage probabilities, even with vague prior specifications. The application to an epileptic treatment dataset, which has been scarcely explored in the literature, resulted in a stable estimation process and successfully linked QOL scores with the time to antiepileptic drug failure. Our proposal is robust to more longitudinal biomarkers, providing more latent information on conditional risk. The survival sub-model could also be extended to multivariate settings such as competing risks or a multi-state outcomes. Furthermore, the inclusion of frailty terms to account for unobserved risk factors, could enhance the model's explanatory power, particularly given the limitations of the dataset available. Another promising extension of our joint model is the development of dynamic prediction scenarios, where the fully parametric approach using defective distributions may assist in forecasting event risk or immunity through longitudinal biomarkers.

\bibliographystyle{apalike}  
\bibliography{references}






\newpage

\appendix
\section{Estimation of individual cure fractions with 95\% credible intervals}
\label{appendix1}


\begin{figure}[htbp]
\centering
\begin{subfigure}[b]{0.45\textwidth}
  \centering
  \includegraphics[width=\linewidth]{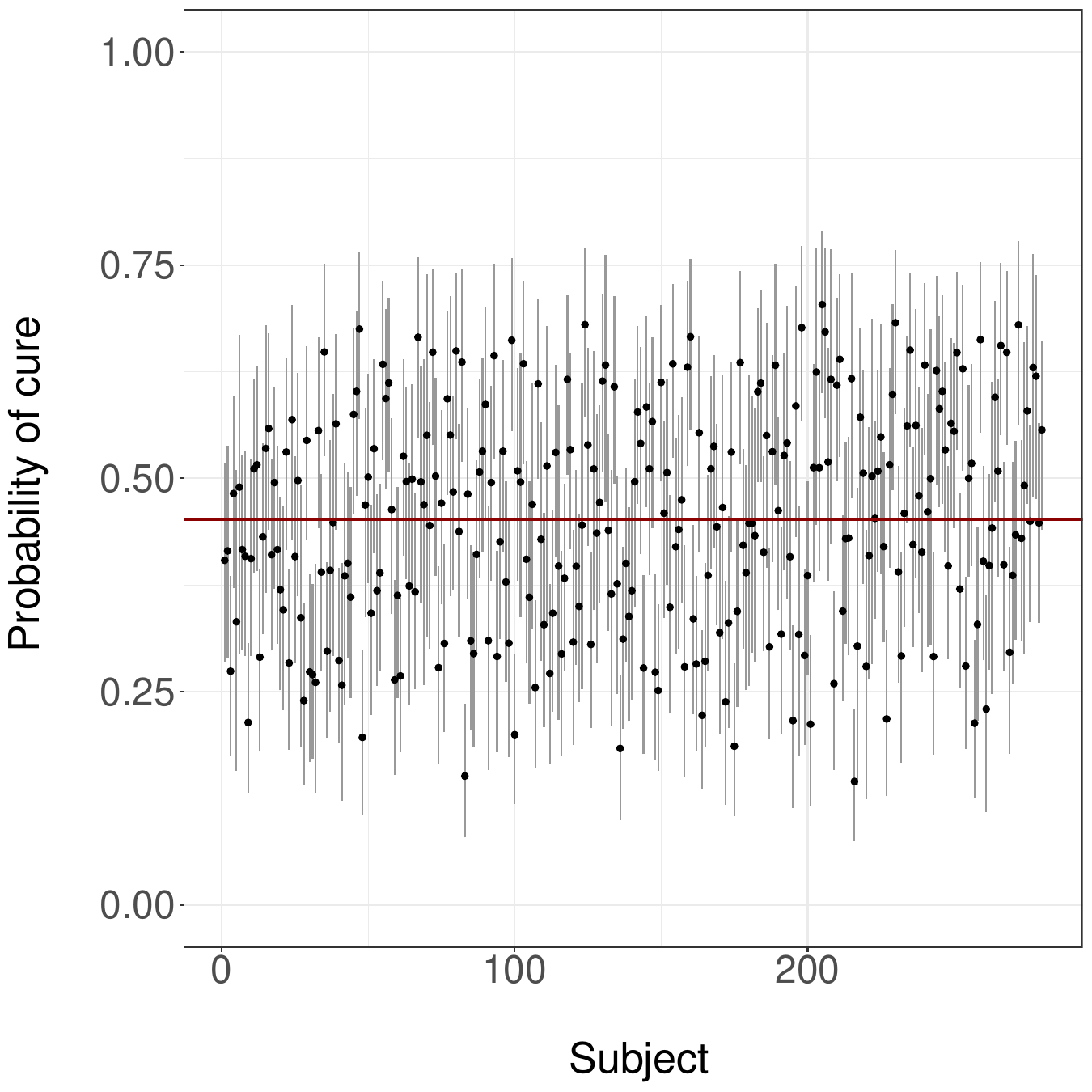}
  \caption{CBZ drug group}
  \label{fig:sub1}
\end{subfigure}
\hspace{0.8 cm}
\begin{subfigure}[b]{0.45\textwidth}
  \centering
  \includegraphics[width=\linewidth]{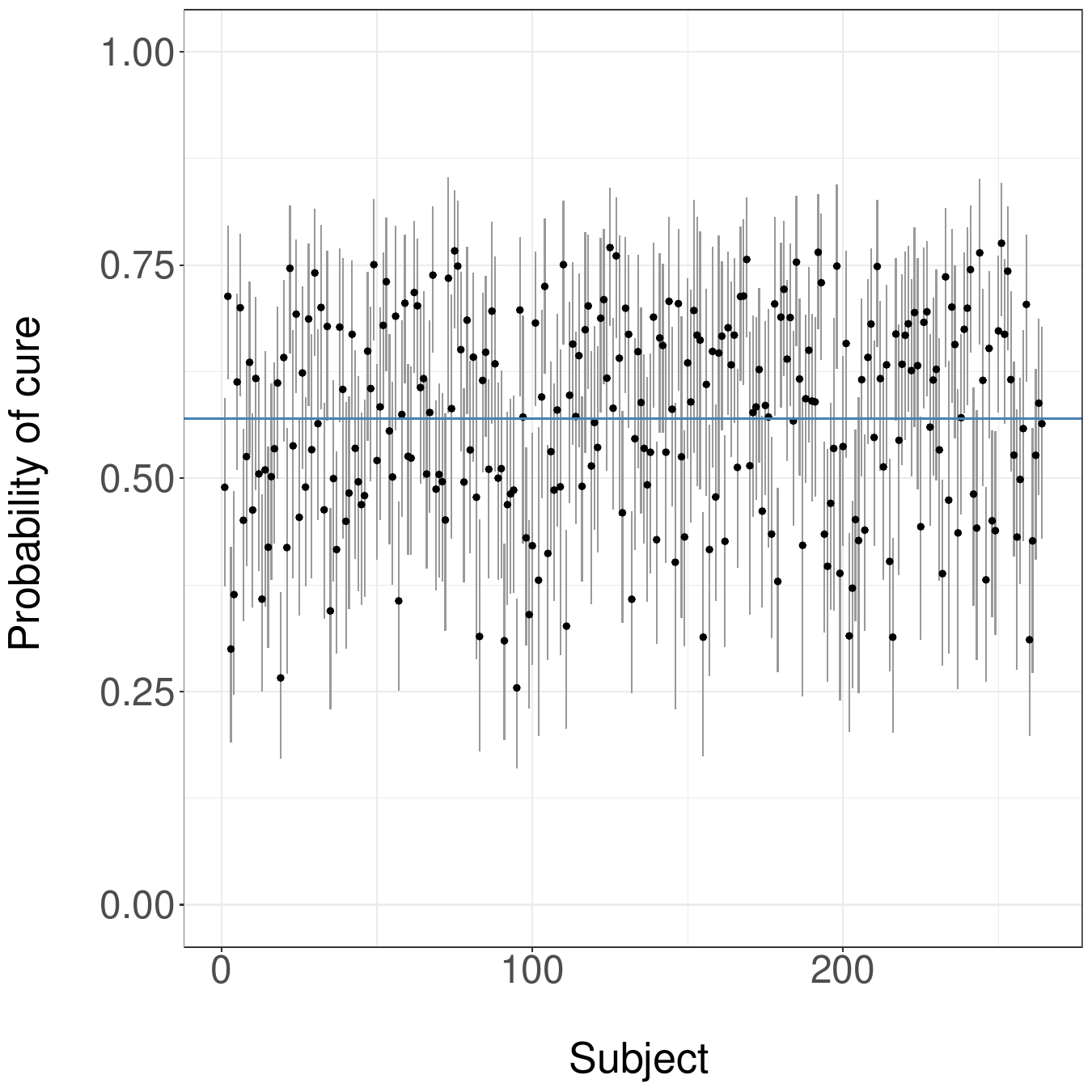}
  \caption{LTG drug group}
  \label{fig:sub2}
\end{subfigure}
\caption{Posterior mean and 95\% credible interval for the individual cure fraction per group of treatment.}
\label{fig:cure-prob-groups}
\end{figure}

\end{document}